\documentclass[runningheads]{llncs}
\usepackage[T1]{fontenc}
\usepackage{graphicx}

\usepackage{amsfonts}
\usepackage{amsmath}
\usepackage{amssymb}

\usepackage{amsthm}
\usepackage{float}
\usepackage{graphics}
\usepackage{hyperref}
\usepackage[svgnames]{xcolor}
\hypersetup{colorlinks={true},urlcolor={blue},linkcolor={DarkBlue},citecolor=[named]{DarkGreen},linktoc=all}

\usepackage{microtype}
\usepackage[capitalise,nameinlink,noabbrev]{cleveref}

\usepackage{doi}

\usepackage{booktabs} 
\usepackage{amsfonts,amsbsy}
\usepackage{thmtools,thm-restate}
\usepackage{enumerate}
\usepackage{bbm}
\usepackage{nicefrac}
\usepackage{wrapfig}
\usepackage[justification = centering]{subcaption}
\usepackage{mathtools}
\usepackage{array,multirow}

\newcommand*\circled[1]{\tikz[baseline=(char.base)]{
            \node[shape=circle,draw,inner sep=2pt] (char) {#1};}}

  {\list{}{\leftmargin=#1\rightmargin=#1}\item[]}%
  {\endlist}

\usepackage{tikz}
\usetikzlibrary{calc,arrows,matrix,decorations.pathreplacing}
\usepackage{paralist}
\usepackage{algorithm}
\usepackage[noend]{algpseudocode}
\usepackage{algorithmicx}

\usepackage[capitalise,noabbrev]{cleveref}
\Crefname{claim}{Claim}{Claims}
\Crefname{corollary}{Corollary}{Corollaries}
\Crefname{definition}{Definition}{Definitions}
\Crefname{example}{Example}{Examples}
\Crefname{lemma}{Lemma}{Lemmas}
\Crefname{property}{Property}{Properties}
\Crefname{proposition}{Proposition}{Propositions}
\Crefname{remark}{Remark}{Remarks}
\Crefname{theorem}{Theorem}{Theorems}
\usepackage{color}

\begin{document}
\title{Envy-Free and Efficient Allocations for Graphical Valuations}
%
%

\author{Neeldhara Misra \and Aditi Sethia}
\authorrunning{N. Misra and A. Sethia}
%
\institute{Indian Institute of Technology, Gandhinagar \email{\{neeldhara.m,aditi.sethia\}@iitgn.ac.in}}
\maketitle              
\begin{abstract}

We consider the complexity of finding envy-free allocations for the class of graphical valuations. 
Graphical valuations were introduced by Christodoulou et al. \cite{10.1145/3580507.3597764} as a structured class of valuations that admit allocations that are envy-free up to any item(EFX). These are valuations where every item is valued by two agents, lending a (simple) graph structure to the utilities, where the agents are vertices and are adjacent if and only if they value a (unique) common item. Finding envy-free allocations for general valuations is known to be computationally intractable even for very special cases: in particular, even for binary valuations, and even for identical valuations with two agents. We show that, for binary graphical valuations, the existence of envy-free allocations can be determined in polynomial time. In contrast, we also show that allowing for even slightly more general utilities $\{0,1,d\}$ leads to intractability even for graphical valuations. This motivates other approaches to tractability, and to that end, we exhibit the fixed-parameter tractability of the problem parameterized by the vertex cover number of the graph when the number of distinct utilities is bounded.
We also show that, all graphical instances that admit EF allocations also admit one that is non-wasteful. Since EFX allocations are possibly wasteful, we also address the question of determining the price of fairness of EFX allocations. We show that the price of EFX with respect to utilitarian welfare is one for binary utilities, but can be arbitrarily large $\{0, 1, d\}$ valuations. We also show the hardness of deciding the existence of an EFX allocation which is also welfare-maximizing and of finding a welfare-maximizing allocation within the set of EFX allocations. 

\end{abstract}

\section{Introduction}

Given a set of $n$ agents and $m$ items, dividing all the items among the agents in a manner considered \textit{fair} for every agent is an important assignment problem. The \textit{fairness} aspect can have different interpretations and has been studied in the literature with various lenses: \textit{envy-freeness \cite{GS58puzzle,F67resource}, equitability \cite{GMT14near}, share-based notions \cite{S48division,B11combinatorial} and their approximations \cite{B11combinatorial,LMM+04approximately,FSV+19equitable,GARG2021103547}}. 

Finding allocations that are envy-free is a gold standard in such assignment problems. This entails that no agent should feel envious of any other agent under the allocation. That is, every agent should value its own allocated bundle at least as much as it values anyone else's bundle. The problem with such envy-free allocations is two-fold: existential and computational. That is, they might not exist for many instances (say, when there are more agents than items), and deciding whether they exist is computationally intractable even for very special and structured instances. In particular, it is NP-complete even for binary valuations (where agents value items at either $0$ or $1$) \cite{AZIZ201571} and weakly NP-Complete for two agents and identical valuations \cite{10.1145/988772.988792}.

Motivated by these issues, we focus on a recently introduced class of structured valuations, called graphical valuations introduced by Christodoulou et al.~\cite{10.1145/3580507.3597764} as a class of valuations that admit allocations that are envy-free up to any item (EFX)\footnote{We refer the reader to next section for the definition of EFX.}. These are valuations where every item is valued by exactly two agents, lending a (simple) graph structure to the utilities, where the agents are associated with vertices and items with edges. Two agent-vertices are adjacent if and only if they value a (unique) common edge-item, represented by the edge between them. 
Such valuations may arise in scenarios where agents only value the items that are geographically closer. For instance, in real estate allocation, potential buyers might only be interested in properties within a certain distance from their workplace or amenities; employees might value office spaces closer to their teams and likewise~\cite{10.1145/3580507.3597764}. 

\paragraph{Our Contributions.} We highlight our main contributions below and put them in context with the already-known results.

\begin{itemize}
    \item We show that an EF allocation if it exists, can be found efficiently for graphical valuations where agents have binary ($\{0,1\}$) valuations over the items (\Cref{thm:EFexists01}). This is in contrast to the intractability of EF allocation for binary utilities in general. 
 \item We show that if we allow for even slightly more general valuations than binary, for instance, $\{0, 1,d\}$-valuations for some constant $d$, the problem again becomes intractable (\Cref{thm:EFHard}). 
 
 \item The above hardness motivates a parameterized approach towards tractability and towards that, we present a \textit{fixed-parameter tractable\footnote{An algorithm that runs in time $f(k)poly(n,m)$ where $f$ is some computable function of the parameter $k$.}} algorithm for finding EF allocations for graphical instances with bounded number of distinct utilities, where the parameterization is in terms of the minimum vertex cover of the associated graph $G$ (\Cref{thm:ILP_VC}).

\item We show that if there is an EF allocation for any graphical instance, then there is also an EF allocation that does not `waste' any item, that is, it does not assign an item to an agent who derives $0$ value from it. This shows that if there is an EF allocation, then there is an EF orientation of the graph $G$ (\Cref{lem:AllocationToOrientation}). This result stands in contrast to the fact that an EFX allocation always exists but an EFX orientation may not exist \cite{10.1145/3580507.3597764}. In terms of the price of EF, this implies that for $\{0, 1\}$-graphical valuations, there is no loss in the welfare while achieving EF allocations, whenever they exist.

\item Christodoulou et al. \cite{10.1145/3580507.3597764} showed that EFX allocations not only always exist but can be found efficiently for graphical valuations. But this comes with a sacrifice in terms of welfare. In particular, there are cases where any EFX allocation must assign items to agents for which they are irrelevant ($0$-valued). In this work, we quantify the loss of welfare while achieving EFX allocations and show that for $\{0,1\}$-graphical instances, the price of EFX for Utilitarian (sum of agents' utilities) welfare is $1$ (\Cref{thm:EFXutil}). That is, restricted to binary graphical valuations, there is no loss in any of the welfare notions and an EFX allocation that maximizes the respective welfare can be found efficiently. On the other hand, we show that for slightly general valuations than binary, that is, for $\{0,1,d\}$-valuations, there are instances with a huge loss in the utilitarian welfare and consequently, price of EFX shoots up to $\infty$ (\Cref{thm:pofinfinity}).

\item On the computational side, we show that for general graphical valuations, finding EFX allocations that also maximize utilitarian welfare is NP-Hard (\Cref{thm:utilhard}). It follows that finding a welfare-maximizing allocation within the set of EFX allocations is also hard.

\end{itemize}

\paragraph{Additional Related Work.} Several special cases and approximations have been extensively studied in the fair division literature to understand the extent of tractability of EFX allocations: binary valuations \cite{10.1007/978-3-031-39344-0_19}; bounded number of agents \cite{doi:10.1137/19M124397X,10.1145/3391403.3399511,10.1145/3580507.3597799}; and bounded number of unallocated items \cite{CGH19envy,Berger_Cohen_Feldman_Fiat_2022}. Graphs have also been associated with fair division in various contexts and models. Allocations, where items allocated to each agent form a connected subgraph in a provided item graph, have been studied \cite{deligkas2021parameterized,BILO2022197}. In a different model, Payan et al. 
\cite{10.5555/3545946.3599032} looked at graph-EFX which requires that an agent, represented by a vertex, satisfy EFX only against its adjacent vertices. Our work is closely aligned with that of Christodoulou et al. \cite{10.1145/3580507.3597764} who introduced graphical valuations and showed the hardness of deciding the existence of an EFX orientation. Following this, Zeng and Mehta \cite{zeng2024structure} characterized that graphs with chromatic number at most $2$ admit EFX orientations for any given valuations, while graphs with chromatic number strictly greater than $3$ may not admit such orientations for all valuations. They also characterized EFX orientability for binary valuations. 

The quantification of welfare loss that is inevitable due to the fairness constraint has also been of interest in the literature. To capture this, the notion of \emph{price of fairness} was proposed in the works of Bertsimas et al. \cite{BFT11price} and Caragiannis et al.~\cite{CKK+12efficiency}. Since then, various works have given bounds for the price of proportionality, envy-freeness, EF1, EFX, equitability, EQ1, maximum Nash welfare, and more \cite{10.1145/2781776,BLM+21price,SCD21connections,SCD23equitability,10.1007/978-3-031-43254-5_16}.

\section{Preliminaries} 
\paragraph{Model.} A fair allocation instance $\mathcal{I}:=(N, M, \mathcal{U})$ consists of a set $N$ of $n$ agents, a set $M$ of $m$ items and a set of valuation functions $\mathcal{U} =\{u_1, \ldots u_n\}$ such that $u_i := 2^M \rightarrow \mathbb{R}_{\geq 0}$. Each $u_i$ captures the utility that agent $i$ derives from a set of items in $M$. A valuation function $u_i$ is said to be additive if the value of a bundle is the sum of the values of the items in the bundle. In this work,  we assume that all the valuation functions are additive. An allocation $\Phi:= \{\Phi_1, \ldots \Phi_n\}$ is a partition of $M$ items into $n$ bundles $\Phi_i$, one for each agent. 

\paragraph{Graphical Allocation Instance.} A graphical fair allocation  instance $\mathcal{I} = \{G = (V, E), \mathcal{U}\}$ takes as input an undirected, simple graph $G$ and a valuation function $\mathcal{U}$. The set of vertices $V$ in $G$ corresponds to $n$ agents and the set of edges $E$ in $G$ corresponds to $m$ items to be allocated. We will often use the terms ``items'' and ``edges'' interchangeably because of this correspondence.
Every agent only values a subset of the incident edge-items. Also, note that every edge is valued by exactly two agents, and every pair of agents value at most one edge together, the one which is incident on both of them. A $\{0,1\}$-graphical instance is one such that $u_i \in \{0, 1\} ~\forall~ i \in N$. Given a graph $G$, an \textit{orientation} $O_G$ is an allocation with the additional property that every
edge is assigned to one of the two endpoints. A directed graph that directs the edges of $G$ towards the vertex that receives the edges is called an \textit{orientation graph} of $G$. Note that every orientation corresponds to a complete allocation. An allocation is an orientation if it assigns the edges to the incident vertices. We say that an orientation satisfies a property if the corresponding allocation satisfies that property.

\paragraph{Fairness Notions.} An allocation $\Phi$ is said to be envy-free (EF) if every agent values its bundle at least as much as it values any other allocated bundle. That is, $u_i(\Phi_i) \geq u_i(\Phi_j)~ \forall~ i, j \in N$. Note that an EF allocation may not exist so we resort to approximations. An allocation is said to be EFX if the envy between a pair of agents can be eliminated by removal of \textit{any} item from the envied bundle. That is, $u_i(\Phi_i) \geq u_i(\Phi_j \setminus \{g\})~ \forall~ g \in \Phi_j~\&~ i, j \in N$. Further, we say that an allocation is non-wasteful if it allocates the items to agents who value them. An item is said to be wasted if it is allocated to an agent who derives $0$ utility from it. Any allocation that consists of wasted items is said to be wasteful.

\paragraph{Welfare Notions and Price of Fairness.} The welfare notion that we consider in this work are defined as follows. Utilitarian social welfare of an allocation $\Phi$ is the sum of the individual agent utilities. We say that an allocation is Utilitarian maximal (UM) if it maximizes the utilitarian welfare. Egalitarian welfare of an allocation $\Phi$ is defined as the minimum utility of any agent
under $\Phi$. We say that an allocation is Egalitarian Maximal (EM) if it maximizes the minimum utility of any agent. Nash welfare is defined as the geometric mean of the utilities of the agents. The price of a fairness notion $F$ with respect to a welfare notion $W$ is the supremum over all fair division instances with $n$ agents and $m$ items of the ratio of maximum welfare under any allocation to the maximum welfare under the fair allocation. In particular, the price of EFX with respect to utilitarian welfare is:



%

 \[PoF_{UM} =  \sup_{I \in \mathcal{I}} \frac{\max_{\Phi^\star \in UM(I)} \sum_i v_i(\Phi^\star_i)}{\max_{\Phi \in EFX(I)} \sum_i v_i(\Phi_i) }\]

\section{Envy-Free Graphical Allocations}


Although it is known that EFX allocations always exist on graphical valuations~\cite{10.1145/3580507.3597764}, an EF allocation may not exist on graphical instances as well, as illustrated by a simple example of a graph consisting of only one edge. Whichever incident vertex receives the edge, the other one is bound to be envious. We show that it is possible to determine if an EF allocation exists in polynomial time for $\{0, 1\}$-graphical valuations, and in the event that the instance admits an EF allocation, such an allocation can be found in polynomial time. Before that, we present a series of structural results. The following result is in contrast to the EFX fairness, where the existence of an EFX allocation does not guarantee an EFX orientation but any EF allocation does guarantee an EF orientation.

\begin{theorem}
\label{lem:AllocationToOrientation}
 Given a graphical allocation instance, there is an EF allocation if and only if there is an EF orientation.
\end{theorem}
\begin{proof} An orientation is EF if the corresponding allocation is EF, so the reverse direction holds. We argue the forward direction.

Suppose there is an EF allocation $\Phi$ for the given instance, which does not correspond to any EF orientation. We assume that everyone values at least one item, otherwise the agent can be removed from the instance. Since $\Phi$ is not an orientation, there must be some edges allocated to vertices that are not incident on them. All such edges are allocated wastefully as an agent does not value an edge that is not incident on itself. Consider the re-allocation $\Phi^\prime$ such that all such wastefully allocated edges are re-allocated to one of their incident vertices, chosen arbitrarily. Say, edge $e = (ij)$ which was previously wastefully allocated to vertex $k$ is now re-allocated to $i$, WLOG. 

Under $\Phi^\prime$, an agent who loses an item can not envy anyone, as its utility does not decrease.
Any agent can potentially be envious of only those agents that are incident on it. Indeed, if $i$ is not incident to agent $k$, then $v_i(\Phi_k) = 0$ as $k$ only receives the edges incident on itself, none of which are valued by $i$. Moreover, suppose $j$ is envious of $i$ under $\Phi^\prime$ as $i$ receives the edge $e =(ij)$ that is also valued by $j$. Notice that $e$ is the only item that is valued by $j$ in the bundle $\Phi_i$ since it is the unique item valued by both $i$ and $j$. Therefore, if $j$ is envious of $i$, we have $u_j(\Phi_i^\prime) = u_j(e) > u_j(\Phi_j^\prime) \geq u_j(\Phi_j)$. The last inequality holds as no agent's utility decreases under the re-allocation $\Phi^\prime$. This implies that $j$ valued $e$ more than the bundle it got under the EF allocation $\Phi$. But then, $u_j(\Phi_j) < u_j(e) \leq u_j(\Phi_k)$, where $k$ is the recipient of $e$ under $\Phi$. This implies that $j$ was envious of $k$ in the allocation $\Phi$, which is a contradiction to the fact that $\Phi$ was EF. Therefore, all the agents are EF under $\Phi^\prime$, and $\Phi^\prime$ assigns edges to only incident vertices. Therefore, $\Phi^\prime$ corresponds to an EF orientation.
\end{proof}

\begin{lemma}
\label{lem:EForientation}
    Given any graphical allocation instance, suppose $v_i^{max}$ is the maximum value any agent $i$ has for any item. Then, a non-wasteful allocation is EF if and only if $i$ gets a utility of at least $v_i^{max} ~\forall~ i \in [n]$.
\end{lemma}

\begin{proof}
    Let $\Phi$ be any EF allocation. Let $v_i^{max}$ be the maximum value an agent $i$ has for an edge $e$. Suppose $v_i(\Phi_i) < v_i^{max}$, then clearly, $e \notin \Phi_i$. Let $e \in \Phi_j$ for some agent $j$. Then $v_i(\Phi_i) < v_i(e) = v_i(\Phi_j)$. Therefore, $i$ is envious of $j$ which is a contradiction to the fact that $\Phi$ is EF. Therefore, every agent $i$ must get a utility of at least $v_i^{max}$ under an EF allocation $\Phi$. Conversely, suppose every agent gets a utility of at least $v_i^{max}$ under a non-wasteful allocation $\Phi$. Since $\Phi$ is a non-wasteful allocation, it corresponds to an orientation in $G$. So every agent receives a subset of edges that are incident on it. Consider an agent $i$. We have $u_i(\Phi_i) \geq v_i^{max}$. Consider any other agent $j$ incident on $i$. If the edge $e =(ij) \in \Phi_j$, then $u_i(\Phi_j) = u_i(e) \leq v_i^{max}$, else $u_i(\Phi_j) = 0$, as $i$ does not value any edge incident on $j$ except $e$. Also, for any agent $j$ not incident on $i$, $u_i(\Phi_j) = 0$ as $i$ does not value any edge which is not incident on itself. Therefore, we have that $u_i(\Phi_i) \geq u_i(\Phi_j)$ for all $1 \leq i \neq j \leq n$ and hence the orientation is EF.
\end{proof}


This gives us the following corollary. 

\begin{corollary}
\label{cor:EFagent} For graphical instances,  if an agent $i$ gets a utility of at least $v_i^{max}$ under a partial orientation $O_P$, then $i$ remains EF under any extension of $O_P$.
\end{corollary}

In particular for binary valuations, there is an EF allocation where every item is allocated to an agent who values it at $1$, so, we have the following result.
\begin{corollary}
    For $\{0,1\}$-graphical instances, the price of EF with respect to utilitarian social welfare is $1$.
\end{corollary}

Note that the above result is not true for binary valuations in general. Consider the instance in \Cref{tab:table1}. It is not a graphical instance as $a_1$ and $a_3$ value $3$ items positively. An EF allocation must allocate at least one item from $\{g_1, g_3, g_4\}$ wastefully. Indeed, if all of them are allocated non-wastefully, then the agent who ends up receiving two of them is envied by the other one. 

\begin{table}
\small
    \centering
    \begin{tabular}{c|cccc}
         & $g_1$ & $g_2$  & $g_3$ & $g_4$ \\ \hline
      $a_1$   & \circled{$1$} & $0$ & $1$ & $1$ \\
       $a_2$  & $0$ & \circled{$1$} & $0$  & \circled{$0$} \\
       $a_3$  & $1$ & $0$ & \circled{$1$}  & $1$ \\
    \end{tabular}
    \caption{An EF allocation that allocates an item wastefully.}
    \label{tab:table1}
\end{table}
 
\begin{theorem}
\label{thm:EFexists01}
    For $\{0, 1\}$-graphical instances, an EF allocation can be found efficiently, if it exists.
\end{theorem}

\begin{proof} Consider an instance $\mathcal{I}$ of $\{0, 1\}$-graphical valuations. Since an EF allocation exists if and only if there is an EF orientation (\Cref{lem:AllocationToOrientation}), we will construct an EF orientation if it exists. For all asymmetric edges $e = (ij)$, we orient them towards the incident agent who values $e$ at $1$, say $i$. This does not create any envy in the graph as the only agent who values $e$ is $i$. We call such vertices $i$ as special vertices since they remain envy-free under any completion of the allocation and are not envied by anyone else. Once we orient all the asymmetric edges, we remove them from the graph. The edges which are valued at $0$ by both end-points are oriented arbitrarily and removed from the graph. This gives us a collection of connected subgraphs $H = \{H_1, H_2, \ldots H_k\}$ such that all edges in $H$ are symmetric and valued at $1$ by both the end-points. For each $H_i \in H$, we consider the following cases:

\begin{enumerate}
\item \label{item:thm2case1} $H_i$ is a tree. Then, $V(H_i) = E(H_i)+1$. By pigeonholing, at least one agent, say $i$, does not receive any edge item from $E(H_i)$. Such a vertex $i$ is always envious under any allocation unless $i$ is already a special vertex. In the former case, there is no complete EF allocation. Otherwise, if there is a special agent $i$, then we root $H_i$ either on $i$ and construct an orientation such that every vertex gets an edge item from its parent. This way, everyone except $i$ receives a utility of at least $1$ from the edges in $H_i$ and hence is EF in any complete orientation. Also, $i$ is EF since it is a special vertex.

\item \label{item:case2} $H_i$ contains a cycle, say $C = \{v_1, v_2, \ldots v_c, v_1\}$. We orient the edges $(v_i, v_{i+1})$ towards $v_i$ and $(v_c, v_1)$ towards $v_c$. Then, every vertex in the cycle is EF as $v_i(\Phi_i) \geq 1$ and remains EF in any completion of this orientation (\Cref{cor:EFagent}). Therefore, the edges inside the cycle can be oriented arbitrarily. 
We now remove the cycle $C$ from $H_i$, replace it with a vertex $c$, and construct a spanning tree of $H_i$ rooted at $c$. We then construct an orientation that allocates every vertex in the spanning tree, except $c$, an edge from its parent. This implies that every agent in the spanning tree except the root $c$ ends up with a utility of at least $1$. All agents corresponding to the root $c$ already had a utility of at least $1$. Since all the agents in $H_i$ now have utility at least $1$, therefore everyone is EF in any completion of the partial orientation. Therefore, the remaining edges in $H_i$ can be oriented arbitrarily, and hence we get an EF allocation for $H_i$.
\end{enumerate}

The algorithm loops over every $H_i$ in $H$ and if there is an EF allocation for every $H_i$, it corresponds to a complete EF allocation (since vertices across components do not envy each other). Else, if there is at least one $H_i$ for which there is no EF allocation, then the algorithm outputs that no complete EF allocation exists. This is true because an envious agent in $H_i$ can not be made EF by any of the edges in the other components, as it does not value them. This settles our claim.
\end{proof}

 We now show in the following result that if we slightly generalize from binary to $\{0,1, d\}$ graphical valuations, it becomes hard to decide if the graphical instance admits an EF allocation.

\begin{theorem}
\label{thm:EFHard}
    Deciding whether an EF allocation exists is NP-Hard even for symmetric $\{0,1,d\}$-graphical valuations.
\end{theorem}

\begin{proof} We present a reduction from~\textsc{Multi-Colored Independent Set (MCIS)}, where given a regular graph $G = (V_1 \uplus \cdots \uplus V_k, E)$, the problem is to decide if there exists a subset $S \subseteq V(G)$ such that $G[S]$ is an independent set and $|V_i \cap S| = 1$ for all $i \in [k]$. We construct the graphical instance as follows.
    All vertices in $V(G)$ correspond to agents and all edges in $E(G)$ to items. Every agent $v \in V(G)$ values its incident edges at $1$. That is, all edges in $G$ are symmetric with a weight of $1$. For every vertex partition $V_i$, we add a vertex-agent $w_i$ adjacent to all the vertices in $V_i$. Every edge $\{(w_i, v): v \in V_i\}$ is a symmetric edge such that $w_i$ and $v$ value it at $d$, where $d$ is the degree of any vertex in the (regular) graph $G$. This completes the construction. A schematic of this construction is shown in \Cref{fig:EFHard}. We now argue the equivalence.

    \begin{figure}
    \centering
    \includegraphics[width=0.5\linewidth]{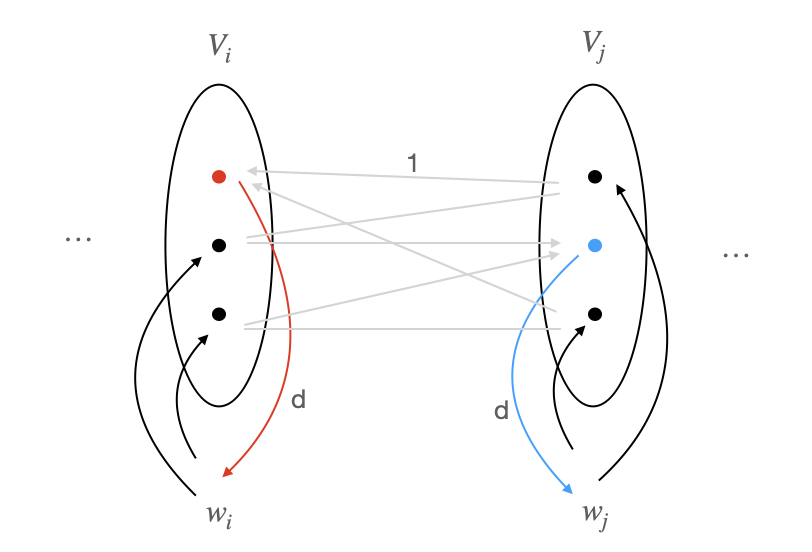}
    \caption{A schematic of reduced instance in the proof of \Cref{thm:EFHard}.}
    \label{fig:EFHard}
\end{figure}

    \paragraph{Forward Direction.}
    Suppose~\textsc{MCIS} is a Yes-instance and there is an independent set $S = \{s_1, s_2, \ldots s_k\} \subseteq V(G)$ such that $G[S]$ is an independent set and $|V_i \cap S| = 1$. Then, we do the following orientation of $E(G)$ to get an EF allocation. 
    \begin{itemize}
    \item $\{(s_i, w_i): i \in [k]\}$ are oriented towards $w_i$.
    \item $\{(v, w_i): v \in V_i \setminus \{s_i\}\}$ are oriented towards $v$.
    \item $\{(s_i, v) : v \in N(s_i) \setminus \{w_i\} \}$ are oriented towards $s_i$.
    \item All the remaining edges are oriented arbitrarily.
    \end{itemize} 

Let $\Phi$ be the allocation corresponding to the above orientation. Then, $u_{w_i}(\Phi_{w_i}) = d$. Note that since all the other edge-items valued at $d$ by $w_i$ are allocated to distinct agents in $V_i$, hence $u_{w_i}(\Phi_j) \leq d$ for any $j \in V(G)$, so all agents $\{w_i : i \in [k]\}$ are envy-free. Similarly all $v \in V_i \setminus \{s_i\}$ have a utility of $d$ each for $\Phi_v$ and a utility of at most $d$ for any other bundle. So, all such agents are envy-free. The remaining agents $\{s_i: i \in [k]\}$ get a utility of $d$ from the $d$ distinct edge-items ($|N(s_i)\setminus w_i| = d$) valued at $1$ each in their respective bundles. Note that all $s_i$ also value any other bundle at atmost $d$ and hence $\{s_i: i \in [k]\}$ are envy-free. This implies that $\Phi$ is an EF allocation.

\paragraph{Reverse Direction.}
Suppose there is an EF allocation $\Phi$ in the reduced instance. Under $\Phi$, each of the $w_i's$ must get at least one incident edge-item to be envy-free. Otherwise, $u_{w_i}(\Phi_i) = 0$ but $w_i$ values every bundle that ends up with any edge-item $\{(w_i, v): v \in V(G)\}$ at $d$, and hence is envious.
Also, since there are only $|V_i|-1$ edge-items valued at $d$ by all the $|V_i|$ agents in $V_i$, so by pigeon-holing, there is at least one agent in every partition $V_i$, say $s_i \in V_i$, which does not end up with a $d$-valued item. Since $u_{s_i}(\Phi_{w_i}) = d$, therefore, $s_i$ must get a utility of at least $d$ from the remaining items in order to be envy-free. This is feasible only if all the agents $\{s_i: i \in [k]\}$ get the respective $d$ edge-items incident on them in the original graph $G$. This implies that $\{s_1, s_2, \ldots s_k\}$ must form an independent set in the original graph $G$. This settles the reverse direction.
\end{proof}

Given the hardness of finding EF allocation for $\{0,1,d\}$-graphical valuations, we consider the parameterized tractability in this context. On a positive note, we show that the problem admits an FPT algorithm parameterized by the vertex cover number of the associated graph, which is the size of the smallest vertex cover (a set of vertices that includes at least one endpoint of every edge) of the graph. We will use the following classical result by Lenstra.

\begin{theorem}[\cite{DBLP:journals/mor/Lenstra83}]
\label{thm:lenstra}
An integer linear programming (ILP) instance of size $L$ with $p$ variables can be solved using $\mathcal{O}\left(p^{2.5 p+o(p)} \cdot\left(L+\log M_{x}\right) \log \left(M_{x} M_{c}\right)\right)$ arithmetic operations and space polynomial in $L+\log M_{x}$, where $M_{x}$ is an upper bound on the absolute value a variable can take in a solution, and $M_{c}$ is the largest absolute value of a coefficient in the vector $c$.
\end{theorem} 

\begin{theorem}
\label{thm:ILP_VC}
    Given a graphical allocation instance with a bounded number of distinct utilities, the problem of finding an EF allocation admits an FPT algorithm parameterized by the Vertex Cover Number of the associated graph $G$.
\end{theorem}

\begin{proof} We formulate an ILP where the number of variables is bounded by a function of the size of the minimum vertex cover of $G$. We will show that the ILP is feasible if and only if there is an EF orientation in the allocation instance. Then, we invoke \Cref{thm:lenstra} to get a feasible solution of the ILP, if it exists, and hence, get the desired FPT algorithm parameterized by minimum vertex cover number.
 
 Let $B$ be the (bounded) set of distinct utilities. Let $S$ be a minimum Vertex Cover of $G$ and $|S| = k$. We have that $I = V(G) \setminus S$ is an independent set. We say that two vertices in $I$ are of the `same class' $C_i$ if they are incident to the same subset of vertices in $S$. This partitions $I$ into at most $2^k$ classes, corresponding to the subsets of $S$. That is, $I = \{C_1, C_2, \ldots C_{2^k}\}$. Further, for each class $C_i$, we say that two vertices have the `same signature' $\sigma_i$ if they value the subset in $S$ in the same manner. That is, $\{v_1, v_2, \ldots v_s\} \in C_i$ have the same signature if their common neighborhood $\{s_1, s_2, \ldots s_t\} \in S$ is valued by all of them at $\{u_1, u_2, \ldots u_t\}$ such that $u_i \in B$. Since the degree of every vertex in $I$ is at most $k$, this gives us at most $|B|^k$ many signatures for every class. All vertices of the same signature in a class are said to be of the same type. In aggregate, we have at most $2^k \cdot |B|^k$ many types of vertices in $I$.

For each vertex $v$ in a type $T$, there are $2^d$ possible orientations of the edges incident on $v$, where $d$ is the degree of $v$ in $G$. Note that $d \leq k$, so there are at most $2^k$ such orientations. We say that an orientation is `good' for the vertex $v$ if it orients at least one of the highest-valued edges of $v$ towards it. We denote the set of good orientations as $O$. 

Towards formulating the ILP, for every type $T$ and a good orientation $o$, we create the variables $x(T, o)$, which denote the number of vertices in the type $T$ that are oriented according to the orientation $o$. Note that these are $f(k) = (2^k \cdot |B|^k) \cdot 2^k = 4^k \cdot |B|^k$ many variables.

We first describe the constraints to ensure the envy-freeness of the vertices in the independent set $I$. Let $n_T$ be the number of vertices in the type $T$. Any vertex is EF if and only if it gets its highest valued edge oriented towards it (\Cref{cor:EFagent}). Therefore, if the vertex $v$ of type $T$ ends up in a good orientation, it is EF. To ensure this, we add the constraints as described in \Cref{eq:EFIS}. Note that LHS of \Cref{eq:EFIS} equals $n_T$ only when every vertex in the type $T$ is oriented according to some good orientation $o$. Indeed, if any vertex fails to end up in a good orientation, then it is not counted in the sum and hence RHS is strictly less than $n_T$, which then violates the constraint.

Now consider a vertex $i$ in $S$.
Let $u(i, T, o)$ denote the utility that an agent $i \in S$ gets when a vertex in type $T$ is oriented according to the orientation $o$. Note that for a fixed orientation, $u(i, T, o)$ is a constant. If $x(T, o)$ many vertices in type $T$ are oriented according to $o$, then the utility that agent $i$ derives from the edge items across $S$ and $I$ from type $T$ is precisely $x(T,o) \cdot u(i, T, o)$.  To capture the utility that $i$ gets from edges in $E(S)$, we can do a brute-force search on which edges are allocated to $i$ (since there are at most ${k \choose 2}$ edges in $E(S)$). To that end, we create binary variables $x_{ie}$ which take value $1$ if the edge $e \in E(S)$ is allocated to $i$, otherwise $0$. These are at most $g(k) = k \cdot k^2$ many variables. And, the utility that $i$ derives from $E(S)$ is precisely $\sum_{e \in E(S)}u_i(e)x_{ie}$. \Cref{eq:eq3} ensures that every edge in $S$ is allocated to at most $1$ agent in $S$, while \Cref{eq:eq4} ensures that every edge in $E(S)$ is allocated. Finally, for $i$ to be EF, it must get at least $v_i^{max}$ utility under any allocation.
This is captured by the constraints in \Cref{eq:EFS}.

\begin{equation}
\footnotesize
\label{eq:EFIS}
    \sum_{o \in O} x(T, o) = n_T ~~~\forall~ T
\end{equation}

\begin{equation}
\footnotesize
\label{eq:eq3}
  \sum_{i \in S} x_{ie} = 1~~~\forall~ e \in E(S)
\end{equation}

\begin{equation}
\footnotesize
\label{eq:eq4}
  \sum_{i \in S} \sum_{e \in E(S)} x_{ie} = |E(S)|
\end{equation}


\begin{equation}
\footnotesize
\label{eq:EFS}
       \sum_{e \in E(S)} u_i(e)x_{ie} + \sum_{T, o} x(T, o) \cdot u(i, T, o) \geq v_i^{max} ~~~\forall~ i \in S
\end{equation}

\begin{equation}
\footnotesize
\label{eq:eq1}
    x(T, o) \geq 0 ~~~\forall~T ~\&~ o \in O
\end{equation}

\begin{equation}
\footnotesize
    x_{ie} \in \{0, 1\} ~~~\forall~i \in S ~\& ~e \in E(S)
\end{equation}

In aggregate, the number of variables created is $f(k) + g(k)$. We now argue the correctness of the ILP. Let $O$ be the orientation that corresponds to the values that $x(T,o)$ takes in some feasible solution of the ILP. For every vertex in the independent set $I$, \Cref{eq:EFIS} ensures that it ends up in a good orientation, and therefore gets one of its highest valued edges oriented towards itself under $O$. This ensures the envy-freeness of vertices in $I$. The envy-freeness of vertices in $S$ is ensured in \Cref{eq:EFS} via a brute-force search of an orientation that gives every vertex in $S$ its highest valued edge. Therefore, a feasible solution to the ILP corresponds to an envy-free allocation of the original instance.

Conversely, suppose there is an EF allocation $\Phi$ in the original instance. Then, by \Cref{lem:AllocationToOrientation}, there is an EF orientation $O$. Let $O_I$ and $O_S$ be the restrictions of $O$ for vertices in $I$ and $S$ respectively. Since $O$ is EF, we have that both $O_I$ and $O_S$ are good orientations. This implies that there exist orientations under which every vertex ends up being in a good orientation. Hence, the constraints \ref{eq:EFIS} and \ref{eq:EFS}, which loop over all the good orientations, are satisfied when the variables $x(T, o)$ and $x_{ie}$ correspond to the orientations $O_I$ and $O_S$ respectively. This implies that the ILP is feasible and this settles our claim.
\end{proof}

\section{EFX and Welfare-Maximization}

In this section, we discuss the price of EFX on graphical instances. Every graph may not admit an EFX orientation but does admit an EFX allocation, so it must be the case that some welfare is lost in the process of achieving EFX. We quantify this loss with respect to Utilitarian welfare.

\begin{theorem}
\label{thm:EFXutil}
    For $\{0, 1\}$-graphical instances, a non-wasteful EFX allocation always exists and can be found in polynomial time. Therefore, the Price of EFX with respect to Utilitarian welfare is $1$.
\end{theorem}

\begin{proof}
    Consider an instance $\mathcal{I}$ of $\{0,1\}$-graphical allocations. For all the asymmetric edges $e =(ij)$, we orient them towards the incident agent who values $e$ at $1$, say $i$. We call $i$ a special vertex, since $i$ is now envy-free in any completion of this partial allocation and is also not envied by anyone else. Once we orient all the asymmetric edges, we remove them from the graph.  The edges that are valued at $0$ by both end-points can be allocated to the non-envied agents arbitrarily at the end of the algorithm, so for now, we remove them from the graph and consider a collection of connected subgraphs $H = \{H_1, H_2, \ldots H_k\}$ such that all edges in $H$ are symmetric and valued at $1$ by both the end-points. For each $H_i \in H$, we consider the following cases:

\begin{enumerate}
\item \label{item:case1EFX} $H_i$ is a Tree. Suppose there is a special agent $i$, then proceed as in the Case \ref{item:thm2case1} of \Cref{thm:EFexists01} and hence arrive at an EF (hence, EFX) orientation on $H_i$.
Suppose there is no special vertex in the tree $H_i$. Then there is no complete EF allocation. To find an EFX allocation, we root $H_i$ on any vertex, say the least degree vertex $i$, and construct an orientation such that every vertex gets an edge item from its parent. Note that $i$ leaves empty-handed and is envious of its neighbors in $H_i$. Since every envied agent (precisely, the children of $i$ in the tree $H_i$) gets exactly one edge item (precisely the edge from $i$), the allocation is EFX.

\item \label{item:case2EFX} $H_i$ contains a cycle, say $C = \{v_1, v_2, \ldots v_c, v_1\}$. This case is the same as Case $\ref{item:case2}$ in the proof of \Cref{thm:EFexists01} and therefore a non-wasteful EF (hence EFX allocation exists) such that every agent receives a utility of at least $1$.
\end{enumerate}

Therefore, we get a non-wasteful EFX allocation $\Phi$. For $\{0, 1\}$-valuations, a non-wasteful allocation is also utilitarian optimal, and hence $\Phi$ is also utilitarian optimal. Therefore, the price of EFX is $1$ in this case.
\end{proof}

\begin{theorem}
\label{thm:pofinfinity}
    The price of EFX with respect to Utilitarian welfare is $\infty$ even for $\{0,1,d\}$-graphical valuations.
\end{theorem}

\begin{proof}
    We construct an instance where the price of fairness is a function of the highest degree of a vertex in the graph. Consider a star graph $G$ rooted at the vertex $r$ which is incident to $d$ many leaf vertices. The root vertex $r$ values each of the $d$ incident edges at $d$. All the leaf vertices value their incident edge at $1$. A utilitarian welfare maximizing allocation gives all the edges to $r$, generating a welfare of $d^2$. Clearly, this allocation is not EFX since the envied agent $r$ has multiple items and every leaf agent violates EFX. Under any EFX allocation, $r$ can not receive more than $1$ item, otherwise, the corresponding leaf vertex whose incident edge is allocated to $r$, violates EFX. Therefore, the maximum welfare under an EFX allocation is $d+(d-1)$, where one $d$-valued edge is allocated to $r$ and the rest all ($d-1$) edges are allocated to their corresponding leaf vertices, valued at $1$ by each of them. Therefore,
    $PoF_{UM} = \frac{d^2}{d+(d-1)} > \frac{d^2}{2d} \approx d. $
    This implies that welfare loss can be as high as possible, and hence PoF is $\infty$. 
\end{proof}


\begin{theorem} 
\label{thm:utilhard} Given an instance of graphical valuations, deciding the existence of a utilitarian welfare-maximizing and EFX allocation (UM+EFX) is NP-Hard.
\end{theorem}

\begin{proof}
    We present a reduction from \textsc{Multi-Colored Independent Set (MCIS)}, where given a regular graph $G= (V_1 \uplus \cdots \uplus V_k, E)$ with degree $d$, the problem is to decide if there exists a subset $S \subseteq V(G)$ such that $G[S]$ is an independent set and $|V_i \cap S| = 1$ for all $i \in [k]$. We construct the graphical instance as follows.
    All vertices in $V(G)$ correspond to agents and all edges in $E(G)$ to items. Every agent $v \in V(G)$ values its incident edges at $1$. That is, all edges in $G$ are symmetric with a weight of $1$. For every vertex partition $V_i$, we add a path of three edges and four vertex-agents $\{w_i^1, w_i^2, w_i^3, w_i^4\}$ such that $w_2^i$ is adjacent to all the vertices in $V_i$. All edges from $w_i^2$ to $V_i$ are valued symmetrically at $d$ by both end-points. The edge ($w_i^1, w_i^2$) is valued at $0$ by $w_i^1$ and at $1$ by $w_i^2$. The edge ($w_i^2, w_i^3$) is valued at $d$ by both $w_i^2$ and $w_i^3$. Finally, the edge ($w_i^3, w_i^4$) is valued at $d$ by $w_i^3$ and at $0$ by $w_i^4$. This completes the construction. A schematic of this construction is shown in \Cref{fig:EFX+UMHard}. We now argue the equivalence.

    \begin{figure}
    \centering
    \includegraphics[width=0.7\linewidth]{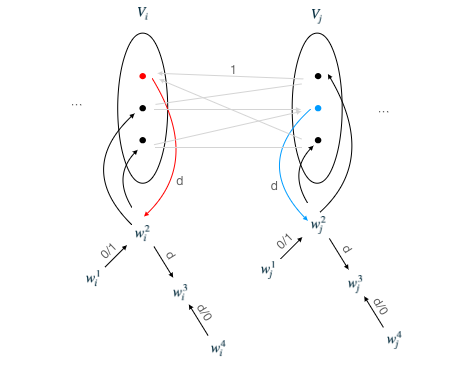}
    \caption{A schematic of reduced instance in the proof of \Cref{thm:utilhard}.}
    \label{fig:EFX+UMHard}
\end{figure}

    \paragraph{Forward Direction.}  Suppose~\textsc{MCIS} is a Yes-instance and there is an independent set $S = \{s_1, s_2, \ldots s_k\} \subseteq V(G)$ such that $|V_i \cap S| = 1$. Then, we do the following orientation of $E(G)$ to get an allocation that is welfare-maximizing and EFX.
    
    \begin{itemize}
    \item $\{(s_i, w_i^2)\}$ are oriented towards $w_i ~\forall~ i \in [k]$.
    \item $\{(v, w_i): v \in V_i \setminus \{s_i\}\}$ are oriented towards $v ~\forall~ i \in [k]$.
    \item $\{(s_i, v) : v \in N(s_i) \setminus \{w_i\} \}$ are oriented towards $s_i~\forall~ i \in [k]$.
    \item $\{w_i^1, w_i^2\}$ are oriented towards $w_i^2~ \forall~  i \in [k]$.
    \item $\{w_i^2, w_i^3\} ~\&~ \{w_i^3, w_i^4\}$ are oriented towards $w_i^3~ \forall~  i \in [k]$.
    \end{itemize} 

Let $\Phi$ be the allocation corresponding to the above orientation. Then, by construction, every edge is allocated to an agent who values it the most. Therefore, $\Phi$ is a (utilitarian) welfare-maximizing allocation. We now argue that $\Phi$ also satisfies EFX. The agents $w_i^1$ and $w_i^4$ do not value any item, so even though they are empty-handed under $\Phi$, they do not envy any agent. All the agents in $V_i$ except $s_i$  get a utility of $d$ each and they value every other bundle at most $d$, hence are envy-free. Likewise, $w_i^2$ is envy-free as it gets a utility of $d+1$ and values every other bundle no more than $d$. Also, $w_i^3$ gets all the edges it values, so there is no envy on its part. Lastly, each $s_i$ gets $d$ of its incident edges valued at $1$ each, deriving a value of $d$, and hence they are envy-free. This implies that $\Phi$ is EF and hence, EFX.

\paragraph{Reverse Direction.} Suppose there is a welfare-maximizing allocation $\Phi$ which also satisfies EFX. Then because $\Phi$ maximizes welfare, it must satisfy the following partial allocation:
$w_i^3$ must receive both its incident edges $\{w_i^2, w_i^3\} ~\&~ \{w_i^3, w_i^4\}$ as it values them highly at $d$, and $\{w_i^1, w_i^2\}$ must be allocated to $w_i^2$ as a utilitarian welfare maximizing allocation is also non-wasteful. This forces $w_i^2$ to be envious of $w_i^3$ even after one item is removed from the envied bundle. Therefore, $w_i^2$ must receive at least one item that it values at $d$ incident to the partition $V_i$. This in turn forces at least one vertex from $V_i$ to violate EFX with respect to $w_i^2$, hence it must receive at least $d$ utility from the remaining items. This is feasible only when it is allocated all its $d$ incident edges. Since this is true for at least one vertex in all $V_i$ such that $i \in [k]$, it must be the case that all these $k$ vertices form an independent set in $G$. This implies that MCIS is a yes-instance. This concludes the argument. 
\end{proof}

We now present a polynomial-time reduction from deciding the existence of a welfare-maximizing and EFX allocation (UM+EFX) to finding a welfare-maximizing allocation within the set of EFX allocations (UM/EFX). Let $w^\star$ be the maximum utilitarian
welfare ($w^\star$ can be computed in linear time by giving
each item to an agent who values it the most). Now suppose the latter problem can be solved in polynomial time. Then, let
$w$ be the maximum welfare within EFX allocations.  If $w=w^\star$, we have a “yes” instance of UM+EFX; else if $w \neq w^\star$, we have a “no”
instance. Therefore, we get the following result.

\begin{corollary}
    Given an instance of graphical valuations, deciding the existence of an EFX allocations with utilitarian welfare at least $w$ (UM/EFX) is NP-Hard.
\end{corollary}

We now discuss the complexity of egalitarian welfare maximizing allocation withing the set of EFX allocations and the loss in the egalitarian welfare due to the EFX constraint.
Recall that the egalitarian welfare of an allocation $\Phi$ is defined as the minimum utility of any agent
under $\Phi$. In the following result, we exhibit the hardness of finding an egalitarian maximizing allocation within the set of EFX allocations (EM/EFX).

\begin{theorem}
\label{thm:egalhard} Given an instance of graphical valuations, deciding the existence of an EFX allocation with egalitarian welfare at least $d$ is NP-Hard.
\end{theorem}

\begin{proof} We show that given a welfare threshold $d$, deciding the existence of an EFX allocation with egalitarian welfare at least $d$ is NP-hard. To that end, we exhibit a reduction from \textsc{Multi-Colored Independent Set (MCIS)}. The construction is similar as in the proof of \Cref{thm:EFHard}.

\paragraph{Forward Direction.}  Suppose~\textsc{MCIS} is a Yes-instance and there is an independent set $S = \{s_1, s_2, \ldots s_k\} \subseteq V(G)$ such that $G[S]$ is an independent set and $|V_i \cap S| = 1$. Then, we do the following orientation of $E(G)$ to get an EF allocation. 
    \begin{itemize}
    \item $\{(s_i, w_i): i \in [k]\}$ are oriented towards $w_i$.
    \item $\{(v, w_i): v \in V_i \setminus \{s_i\}\}$ are oriented towards $v$.
    \item $\{(s_i, v) : v \in N(s_i) \setminus \{w_i\} \}$ are oriented towards $s_i$.
    \item All the remaining edges are oriented arbitrarily.
    \end{itemize} 

Let $\Phi$ be the allocation corresponding to the above orientation. Then, for all agents $a$, we have $u_{a}(\Phi_{a}) = d$. By \Cref{lem:EForientation}, $\Phi$ is an EF (hence, EFX) allocation. Therefore, there is an EFX allocation with egalitarian welfare of at least $d$.

\paragraph{Reverse Direction.} Suppose there is EFX allocation with egalitarian welfare at least $d$. Then, every $w_i$ must receive at least one of its incident edges to derive a utility of at least $d$. This implies that there is at least one vertex $v$ in every $V_i$ does not get the edge $(v, w_i)$. In order to get a utility of at least $d$, each of such vertices $v$ must get all of its $d$ incident edges in $G$. This is feasible if and only if they form an independent set in $G$. Therefore, if there is EFX allocation with egalitarian welfare at least $d$, then MCIS is a yes-instance and this settles the claim.
\end{proof}

\begin{corollary}
\label{cor:Egalone}
 For $\{0, 1\}$-graphical instances, the Price of EFX with respect to Egalitarian welfare is $1$.
\end{corollary}

\begin{proof}
    If the optimal egalitarian welfare is $0$, then there is nothing to prove. Otherwise, consider the case when the optimal egalitarian welfare is $k>0$ and is achieved by the allocation $\Phi$. We have $u_i(\Phi_i) \geq k ~\forall~ i \in [n]$. If $\Phi$ is a wasteful allocation, then consider the re-allocation of all wastefully allocated edges to one of its incident vertices, arbitrarily. Since this re-allocation does not bring down the utility of any agent, it remains egalitarian optimal but now corresponds to a non-wasteful allocation and hence, an orientation. By \Cref{lem:EForientation}, we get that that this orientation is EF, since $v_i(\Phi_i) = k \geq 1 = v_i^{max} ~\forall~i \in [n]$. Therefore, we get an egalitarian optimal EF (hence, EFX) allocation and this settles our claim.
    \end{proof}

Next, we show that for binary graphical valuations, there is no loss in the Nash welfare as well. 

\begin{corollary}
\label{cor:Nashone}
 For $\{0, 1\}$-graphical instances, the Price of EFX with respect to Nash welfare is $1$.
\end{corollary}

\begin{proof}
    If the optimal Nash welfare is $0$, then there is nothing to prove. Otherwise, consider the case when the optimal Nash welfare is $k>0$ and is achieved by the allocation $\Phi$. Then, we must have $u_i(\Phi_i) \geq 1 ~\forall~ i \in [n]$. Since every Nash optimal allocation is non-wasteful, therefore, $\Phi$ is a non-wasteful allocation and hence, an orientation. By \Cref{lem:EForientation}, we get that that this orientation is EF, since $v_i(\Phi_i) \geq 1 = v_i^{max} ~\forall~i \in [n]$. Therefore, we get a Nash optimal EF (hence, EFX) allocation and this settles our claim.
\end{proof}

Since Nash optimal allocations can be found in polynomial time for general 
binary additive valuations \cite{DARMANN2015548}, we can find an EFX and Nash welfare maximizing allocation in polynomial time. 

We say that an allocation is a leximin allocation if, among all allocations, it lexicographically maximizes the utility profile, that is, maximizes the minimum utility, subject to that maximizes the
second minimum, and so on. Clearly, leximin allocations are also egalitarian maximal allocations. Moreover, for general binary additive valuations, the set of leximin and Nash optimal allocations coincide \cite{10.1007/978-3-030-64946-3_26}, and hence a Nash optimal allocation is also maximizes egalitarian welfare. Since Nash optimal allocations are also non-wasteful, they also maximize the utilitarian welfare. This gives us the following result.

\begin{corollary}
    For $\{0,1\}$-graphical instances, an EFX allocation that maximizes Utilitarian, Egalitarian, and Nash welfare always exists and can be found in polynomial time.
\end{corollary}


\section{Concluding Remarks.} We studied the complexity of finding envy-free allocations for graphical valuation and quantified the loss of welfare in the process of achieving approximate (i.e, EFX) envy-freeness, which was the original motivation for the study of the class of graphical valuations. We believe there are several directions of interest for future work that build on our preliminary line of inquiry here. For instance, for parameterized results, one could consider structural parameters that are smaller than vertex cover. Extending the PoF discussion beyond binary setting for other welfare notions would also be of interest. Finally, one might also generalize the class of graphical valuations in many ways. One generalization is to allow graphs with multiedges which then corresponds to instances where an item is liked by at most two agents but a pair of agents together can derive positive value from more than one item. The restricted setting of hypergraphs where every edge corresponds to multiple but the same number of vertices is also an interesting direction to pursue.



%
%
%
%
\bibliographystyle{splncs04}
\bibliography{main.bib}
%






\end{document}